\documentclass[11pt]{llncs}
\usepackage{algorithm2e}
\usepackage{algorithmicx}
\usepackage{epsfig}
\usepackage{fancybox}
\usepackage{graphics}
\usepackage{geometry}

\begin{document}
\title{Generalized $\beta$-skeletons
\thanks{This research is supported by the ESF EUROCORES programme EUROGIGA,
CRP VORONOI.}}
\author{
        Miros{\l}aw Kowaluk  
            \and
        Gabriela Majewska  
\institute{
Institute of Informatics, University of Warsaw, Warsaw, Poland,
\texttt{kowaluk@mimum.edu.pl}
\texttt{gm248309@students.mimuw.edu.pl}
}}

\maketitle

\begin{abstract}
$\beta$-skeletons, a prominent member of the neighborhood graph family, have
interesting geometric properties and various applications ranging from
geographic networks to archeology.
This paper focuses on developing a new, more general than the present one, definition
of $\beta$-skeletons based only on the distance criterion.
It allows us to consider them in many different cases, e.g. for weighted graphs 
or objects other than points.
Two types of $\beta$-skeletons are especially well-known: the Gabriel Graph 
(for $\beta = 1$) and the Relative Neighborhood Graph (for $\beta = 2$).
The new definition retains relations between those graphs and the other well-known ones
(minimum spanning tree and Delaunay triangulation).
We also show several new algorithms finding $\beta$-skeletons.           
\end{abstract}

\section{Introduction}

The $\beta$-skeletons \cite{kr85} in $R^2$ belong to the family of
proximity graphs, geometric graphs in which two vertices (points) are connected
with an edge if and only if they satisfy particular geometric requirements.
In the case of $\beta$-skeletons, those requirements are dependent on a given parameter 
$\beta \geq 0$. 
 
The $\beta$-skeletons are both important and popular because of many practical 
applications. The applications span a wide spectrum of areas: 
from geographic information systems to wireless ad hoc networks and machine learning. 
They also facilitate reconstructing  shapes of two-dimension objects from 
sample points, and are also useful in finding the minimum weight triangulation 
of point sets. 
  
{\em Gabriel graphs} ($1$-skeletons), 
defined by Gabriel and Sokal \cite{gs69}, are an example of $\beta$-skeletons for $\beta=1$.

The {\em relative neighborhood graph} (\textit{RNG}) is  another example of the $\beta$-skeleton
graph family, for $\beta=2$. The \textit{RNG} was 
introduced by Toussaint \cite{tou80} in the context of their applications in pattern 
recognition. 

Kirkpatrick and Radke \cite{kr85} proved an important theorem connecting $\beta$-skeletons 
with the minimum spanning tree $\mathit{MST}(V)$ and Delaunay triangulation $\mathit{DT}(V)$ of $V$ :

\begin{theorem}
\label{faktinkluzja}
Let us assume that points in $V$ are in general position.
For $1\leq \beta \leq \beta' \leq 2$ following inclusions hold true: 
$\mathit{MST}(V) \subseteq \mathit{RNG}(V) \subseteq G_{\beta'}(V)\subseteq G_{\beta}(V) \subseteq GG(V) 
\subseteq \mathit{DT}(V)$.
\end{theorem}

According to this theorem, many algorithms computing Gabriel graphs and relative neighborhood graphs
in subquadratic time were created \cite{ms84,su83,jk87,jky89,l94}.

There are very few algorithms for Gabriel graphs and relative neighborhood graphs
in metrics different than the euclidean one (\cite{mg11,w06}). In particular, $\beta$-skeletons 
in different metrics have not been studied and this paper makes an initial effort to fill this gap.
Our main purpose of this paper is to develop a definition of $\beta$-skeleton based 
only on the distance criterion.
Moreover, $\beta$-skeletons defined in a new way fulfill all conditions of Theorem \ref{faktinkluzja}.

Two different forms of $\beta$-neighborhoods have been studied for $\beta > 1$
(see for example  \cite{abe98,e02})   
leading to two different families of $\beta$-skeletons: lens-based $\beta$-skeletons 
and circle-based $\beta$-skeletons. 
The new definition of $\beta$-skeletons can be used in both cases. 
However, in this work, we focus on the lens-based $\beta$-skeletons.

The paper is organized as follows.
The definition and basic properties of $\beta$-skeletons are introduced in Section 2. 
In Section 3 we present a distance based definition of a $\beta$-skeleton in $l_p$ metrics. 
In Section 4 we discuss 
a problem of not uniquely defined centers of discs determining regions of a $\beta$-skeleton. 
Then, we consider $\beta$-skeletons in weighted graphs. In the next section we describe a case 
when the generators of the regions are not uniquely defined. In Section 7 we formulate 
the definition of the generalized $\beta$-skeletons. In Section 8 we present some algorithms 
computing generalized $\beta$-skeletons for cases discussed in the previous sections. 
The last section contains open problems and conclusions.

\section{Preliminaries}

We consider a two-dimensional plane $R^2$ with the $l_p$ metric (with distance 
function $d_p$), where $1<p<\infty$. 

\begin{definition} \cite{kr85}
\label{betaskeleton}
For a given set of points  $V=\{v_1, v_2, \dots , v_n\}$ in $R^2$ and parameters 
$\beta \geq 0$ and $p$ we define graph 
$G_{\beta}(V)$ -- called a lens-based $\beta$-skeleton -- as follows:  
two points $v_1, v_2$ are connected with an edge if and only if no point 
from $V \setminus \{v_1, v_2\}$ belongs to the set 
$N_{p}(v_1, v_2,\beta)$ where:

\begin{enumerate} 
\item 
for $\beta=0$,  $N_{p}(v_1, v_2,\beta)$ is equal to the segment $v_1v_2$;
\item 
for $0<\beta<1$, $N_{p}(v_1, v_2,\beta)$ is the intersection of two discs in $l_p$, 
each of them has
radius $\frac{d_p(v_1v_2)}{2\beta}$ and whose boundaries contain both $v_1$ and $v_2$;
\item 
for $1 \leq \beta<\infty$, $N_{p}(v_1,v_2,\beta)$ is the intersection of two $l_p$ 
discs, each with radius $\frac{\beta d_p(v_1v_2)}{2}$, whose centers are in points 
$(\frac{\beta}{2})v_1+(1-\frac{\beta}{2})v_2$ and 
in $(1-\frac{\beta}{2})v_1+(\frac{\beta}{2})v_2$, respectively;
\item 
for $\beta=\infty$, $N_{p}(v_1,v_2,\beta)$ is the unbounded strip between lines 
perpendicular to the segment $v_1v_2$ and containing $v_1$ and $v_2$.
\end{enumerate}

The region $N_{p}(v_1, v_2,\beta)$ is called a lens (the name of the skeleton type, used 
in the literature, i.e. lune-based, is descended from the complement of the lens to the disc 
which looks like a lune - see Figure \ref{fig:region}). 
\end{definition}

Furthermore,  we can  consider  open or closed $N_{p}(v_1,v_2,\beta)$ regions
that lead to {\em an open} or {\em closed $\beta$-skeletons}. For example,  the
{\em Gabriel graph} is the closed $1$-skeleton and the
{\em relative neighborhood graph} is the open $2$-skeleton.

In a similar way we can define a different family of graphs by changing 
the lens-based definition for $\beta > 1$. This new family of graphs is called 
circle-based $\beta$-skeletons. \\ 

\begin{definition} \cite{e02}
For a given set of points  $V=\{v_1, v_2, \dots , v_n\}$ in $R^2$ and parameters 
$\beta \geq 0$ and $p$ we define graph 
$G^c_{\beta}(V)$ -- called a circle-based $\beta$-skeleton -- as follows:  
two points $v_1, v_2$ are connected with an edge if and only if no point 
from $V \setminus \{v_1, v_2\}$ belongs to the set 
$N^c_{p}(v_1, v_2,\beta)$ where:
\begin{enumerate}
\item
for $\beta <1$ we define set $N^c_p(v_1,v_2,\beta)$ the same way as for lens-based $\beta$-skeleton;
\item
for $1 \leq \beta<\infty$ the set $N^c_{p}(v_1,v_2,\beta)$ is a union of two discs in $l_p$, 
each with radius $\frac{\beta d_p(v_1v_2)}{2}$, whose boundaries contain both $v_1$ and $v_2$;
\item
for $\beta=\infty$, $N^c_{p}(v_1,v_2,\beta)$ is a union of the segment $v_1v_2$ and two open 
hyperplanes defined by the line passing by $v_1$ and $v_2$.
\end{enumerate} 
\end{definition}

In this work we generally focus on lens-based $\beta$-skeletons.\\ 

\begin{figure}[htbp]
\centering
\includegraphics[scale=0.3]{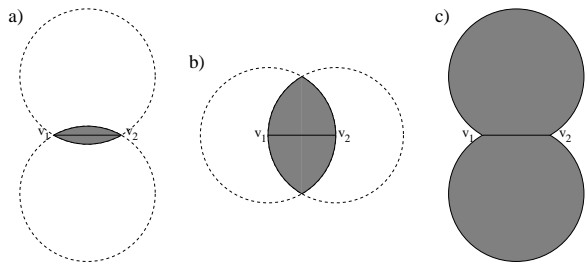}
\caption{Lenses of the $\beta$-skeleton for $0 < \beta <1$ (a), the lens-based skeleton 
(b) and the circle-based skeleton (c) for $1 \leq \beta \leq \infty$.}
\label{fig:region}
\end{figure}



\section{Distance based definition of $\beta$-skeletons }
\label{distancebased}


To describe $\beta$-skeletons for $\beta \in \{0, \infty\}$ in the remaining part of this paper 
we will use the definition of Painleve-Kuratowski convergence \cite{k66}: 

\begin{definition}
\label{kuratowskiconv}
Let $X$ be a space with metric $d$, $ \{ S_n \}_{ n \in N}$ be a sequence 
of subsets in $X$ and for any $x \in X$ there is $d(x,S_n)=inf\{d(x,s)|s \in S_n\}$. Then:
\begin{enumerate}
\item
the upper (outer) limit $\limsup_{n \rightarrow \infty} S_n$ of $\{S_n\}$ 
is defined as a such a set that:\\
$\limsup_{n \rightarrow \infty} S_n=\{p \in X| \liminf d(p, S_n)=0\}$;
\item
the lower (inner) limit 
$\limsup_{n \rightarrow \infty} S_n$ of $\{S_n\}$ is such a set that:\\
$\liminf_{n \rightarrow \infty} S_n=\{p \in X| \limsup d(p, S_n)=0 \}$;
The sequence $ \{ S_n \}_{ n \in N}$ is said to be convergent in the sense of Painleve-Kuratowski
and denoted as $\lim_{n \rightarrow \infty} S_n =S$ 
if $\limsup_{n \rightarrow \infty} S_n=\liminf_{n \rightarrow \infty} S_n = S$.
\end{enumerate}
\end{definition}




The regions $N_{p}(v_1,v_2,\beta)$ are uniquely defined in $l_p$ metric for each $0 < \beta < \infty$.
Therefore $\lim_{\beta \rightarrow 0} N_p(v_1,v_2, \beta)$ and 
$\lim_{\beta \rightarrow \infty} N_p(v_1,v_2, \beta)$ are also unique and
for any two points $v_1,v_2 \in V$ we have 
$\lim_{\beta \rightarrow 0} N_p(v_1,v_2, \beta)=N_p(v_1,v_2,0)$ and 
$\lim_{\beta \rightarrow \infty} N_p(v_1,v_2, \beta)=N_p(v_1,v_2, \infty)$.

Let us consider lens-based $\beta$-skeletons in $l_p$ metric, where $\beta \geq 1$. 
Let $c_1$ ($c_2$, respectively) be the center of a disc defining the lens $N_{p}(v_1, v_2,\beta)$ 
and $d_p(v_1, c_1)=\frac{\beta d_p(v_1v_2)}{2}$  ($d_p(v_2, c_2)=\frac{\beta d_p(v_1v_2)}{2}$, 
respectively). 
It is easy to notice that $d_p(c_1, c_2)=(\beta -1)d_p(v_1,v_2)$ 
and $d_p(v_1, c_2)=d_p(v_2,c_1)=|\frac{(\beta -2) d_p(v_1v_2)}{2}|$ .

\begin{lemma}
\label{lemata}
Let $c_1, c_2$ be points in $R^2$ with a $l_p$ metric, where $1<p<\infty$,
$d_p(c_1, c_2)=(\beta -1)d_p(v_1,v_2)$ and 
$d_p(v_1, c_2)=d_p(v_2,c_1)=|\frac{(\beta -2) d_p(v_1v_2)}{2}|$ . 
Then, the points $c_1, c_2$ are centers of discs determining the lens $N_{p}(v_1, v_2,\beta)$.
\end{lemma}

\begin{proof}
For $1\leq \beta < 2$ we have
$d_p(v_2, c_1)+d_p(c_1, c_2)+d_p(v_1, c_2)= 
2|\frac{(\beta -2) d_p(v_1v_2)}{2}|+(\beta -1)d_p(v_1,v_2)=
(2-\beta)d_p(v_1,v_2)+(\beta -1)d_p(v_1,v_2)=d_p(v_1, v_2)$.
For $2 \leq \beta<\infty$ we have
$d_p(v_2, c_1)+d_p(v_1, v_2)+d_p(v_1, C_2)= 
2|\frac{(\beta -2) d_p(v_1v_2)}{2}|+d_p(v_1,v_2)=
(\beta-2)d_p(v_1,v_2)+d_p(v_1,v_2)=(\beta -1)d_p(v_1, v_2)=d_p(c_1, c_2)$.
Hence, the points $v_1, v_2, c_1, c_2$ are colinear.
Moreover, from the triangle inequality it follows that
in an arbitrary metric for $1 \leq \beta < 2$ is 
$\frac{\beta d_p(v_1v_2)}{2}=d_p(v_1, c_2)+d_p(c_1, c_2) \geq d_p(v_1, c_1) \geq 
d_p(v_1, v_2)-d_p(v_2, c_1)=\frac{\beta d_p(v_1v_2)}{2}$
and for $2 \leq \beta < \infty$ is
$\frac{\beta d_p(v_1v_2)}{2}=d_p(v_2, c_1)+d_p(v_1, v_2) \geq d_p(v_1, c_1) \geq 
d_p(c_1, c_2)-d_p(v_1, c_2)=\frac{\beta d_p(v_1v_2)}{2}$
The same relations occur for $d_p(v_2, c_2)$. 
Hence $d_p(v_1, c_1)=d_p(v_2, c_2)=\frac{\beta d_p(v_1v_2)}{2}$ (discs defining the lens
have the radius $\frac{\beta d_p(v_1v_2)}{2}$), i.e. the lens is $N_p(v_1,v_2, \beta)$.
\end{proof}

\begin{corollary}
\label{lp-distance}
Distance conditions in Definition \ref{betaskeleton} for $0<\beta<1$ and in Lemma \ref{lemata}
for $1 \leq \beta <\infty$ correctly define lenses of $\beta$-skeletons for $0 < \beta <\infty$.
The Painleve-Kuratowski convergence correctly defines lenses of $\beta$-skeletons for $\beta=0$ 
and $\beta=\infty$. 
\end{corollary}

\section{Not uniquely defined centers of the discs determining lens}
\label{l1-section}

Note that for metrics $l_1$ and $l_{\infty}$ corollary \ref{lp-distance} is not true,
since centers of discs determining a lens can be defined not uniquely.
Let us assume that points belonging to $V$ are in general position.
For $1 \leq \beta < \infty$, intersection of circles centered in $v_1$ and $v_2$ with
radiuses $|\frac{(\beta -2) d_p(v_1v_2)}{2}|$ and $\frac{\beta d_p(v_1v_2)}{2}$,
respectively, contains may points. 
If a line passing through $v_1$ and $v_2$ is parallel to sides of squares being discs 
in $l_1$ or $l_{\infty}$ then centers of discs determining a lens for $0 < \beta < 1$ 
are not uniquely defined too.

We will concentrate on $l_1$ metric and describe what should be modified in the definition 
of the $\beta$-skeleton to be able to use it also in this metric. The consideration concerning
$l_{\infty}$ metric are exactly the same. 

In this case we analyze lens defined by discs with centers satisfying the conditions below.

For $v_1, v_2 \in V$ and $0<\beta<1$, the points $c_1$ and $c_2$ are centers of discs defining
$N_1(v_1, v_2,\beta)$ when 
$d_1(c_1,v_1)=d_1(c_2,v_1)=d_1(c_1,v_2)=d_1(c_2,v_2)=\frac{d_1(v_1,v_2)}{2 \beta}$ 
and each shortest path connecting $c_1$ and $c_2$ intersects some shortest path between 
$v_1$ and $v_2$ (the second condition guarantees that $v_1$ and $v_2$
are ends of the lens just like in $l_p$). 
Note that the shortest path between two points in $l_1$ does not have to be a segment.

For $1 \leq \beta < \infty$ the points $c_1$ and $c_2$ are centers of discs defining a lens
$N_1(v_1, v_2,\beta)$ when $d_p(c_1, c_2)=(\beta -1)d_p(v_1,v_2)$ and 
$d_p(v_1, c_2)=d_p(v_2,c_1)=|\frac{(\beta -2) d_p(v_1v_2)}{2}|$ .

From now on we will alternately use a symbol $N_1(v_1, v_2,\beta)$ for description
of concrete lens (one element set) or for a set of lenses defined for an edge $v_1v_2$. 

\begin{lemma}
\label{l1-01}
For $v_1,v_2 \in V$ and for $0<\beta<\beta'< 1$ there is $N_1(v_1, v_2,\beta)=N_1(v_1, v_2,\beta')$.
\end{lemma}
\begin{proof}
Note that if points $v_1$ and $v_2$ are not on a line parallel to sides of a square being a circle 
in $l_1$ then points $c_1$ and $c_2$ are defined uniquely and the lens $N_1(v_1, v_2,\beta)$ 
is a rectangle with opposite vertices in points $v_1, v_2$ and sides parallel to sides of a square 
being a circle in $l_1$ (see Figure \ref{fig:l_1}).\\
In other case, there are more than two points belonging to the intersection of discs with radius 
$\frac{d_1(v_1,v_2)}{2 \beta}$ and centers in points $v_1$ and $v_2$. 
However, there are only two pairs of points $c_1, c_2$ for which each shortest path connecting 
$c_1$ and $c_2$ intersects some shortest path between $v_1$ and $v_2$. Both pairs define 
the same lens and $N_1(v_1, v_2,\beta)=v_1v_2$.   
\end{proof}

Note that $N_1(v_1, v_2,0)$ is uniquely defined for any $v_1,v_2$ 
and $N_1(v_1, v_2,0)=N_1(v_1, v_2,\beta)$ for $0<\beta<1$.\\

\begin{lemma}
\label{l1-betabig}
For any two points $v_1,v_2 \in V$ a lens $N_1(v_1,v_2,2)$ is uniquely defined. 
For $\beta > 2$ there exists a lens $N \in N_1(v_1,v_2,\beta)$ such that $N=N_1(v_1,v_2,2)$. 
Moreover,  $N_1(v_1,v_2,2)$ is contained in every lens in $N_1(v_1,v_2,\beta)$.
\end{lemma}
\begin{proof} 
For $\beta=2$ we have $c_1=v_2$ and $c_2=v_1$.
Let $y(v)$ ($x(v)$, respectively) denote a value of $y$-coordinate ($x$-coordinate, respectively)
of a point $v$. Let us assume that $|y(v_1)-y(v_2)| \leq |x(v_1)-x(v_2)|$.
It is easy to notice that for $\beta > 2$ and centers $c_1, c_2$ of discs defining a lens 
$N \in N_1(v_1,v_2,\beta)$ such that
$y(c_1)=y(v_1)$ and $y(c_2)=y(v_1)$ we have $N=N_1(v_1,v_2,2)$.
Lenses in $N_1(v_1,v_2,\beta)$ are an rectangles. Points $v_1$ and $v_2$ lie on their opposite sides
(see Figure \ref{fig:l_1}). Since $d_p(c_1, c_2)=(\beta -1)d_p(v_1,v_2)$ we have 
$|y(c_1)-y(c_2)| \geq |y(v_1)-y(v_2)|$. If the difference $|y(c_1)-y(c_2)|$ grows then 
the opposite vertices of the lens drive away. Hence $N_1(v_1,v_2,2)$ is contained in each lens 
from $N_1(v_1,v_2,\beta)$. A similar argumentation in respect of $x$-coordinates is true when 
$|x(v_1)-x(v_2)|<|y(v_1)-y(v_2)|$.        
\end{proof} 

\begin{lemma}
\label{l1-twolenses}
For $v_1,v_2 \in V$ and $1 < \beta < 2$ all lenses from $N_1(v_1,v_2,\beta)$ are
included in $N_1(v_1,v_2,2)$.
For $1 < \beta < 2$ there exist two lenses $N, N' \in N_1(v_1,v_2,\beta)$ such that each of them 
does not contain any lens from $N_1(v_1,v_2,\beta) \setminus \{N, N'\}$, $N \cup N' = N_1(v_1,v_2,2)$
and $N \cap N' \neq \emptyset$. 
\end{lemma}
\begin{proof}
Let us assume that $|y(v_1)-y(v_2)| \leq |x(v_1)-x(v_2)|$ and $x(v_1)<x(v_2), y(v_1)<y(v_2)$.
Centers of discs determining lenses in $N_1(v_1,v_2,\beta)$ belong to the rectangle 
$[x(v_1),x(v_2)] \times [y(v_1),y(v_2)]$. Therefore each such a disc is contained in one of the discs
determining $N_1(v_1,v_2,2)$. Hence, each lens from $N_1(v_1,v_2,\beta)$ is included 
in $N_1(v_1,v_2,2)$. 
Let $c_1, c_2$ be centers of discs determining a lens in $N_1(v_1,v_2,\beta)$. The lens has
minimum size when $x(c_1)=x(c_2)$ or $y(c_1)=y(c_2)$. Those lenses are in extreme positions
when $y(c_1)=y(c_2)=y(v_1)$ or $y(c_1)=y(c_2)=y(v_2)$. Three sides of the first and the second
lens are contained in sides of $N_1(v_1,v_2,2)$. Both lenses contain vertices $v_1$ and $v_2$.
Hence, their intersection is not empty and their sum is $N_1(v_1,v_2,2)$. 
A similar argumentation in respect of $x$-coordinates is true when $|x(v_1)-x(v_2)|<|y(v_1)-y(v_2)|$.
\end{proof}

\begin{figure}[htbp]
\centering
\includegraphics[scale=0.3]{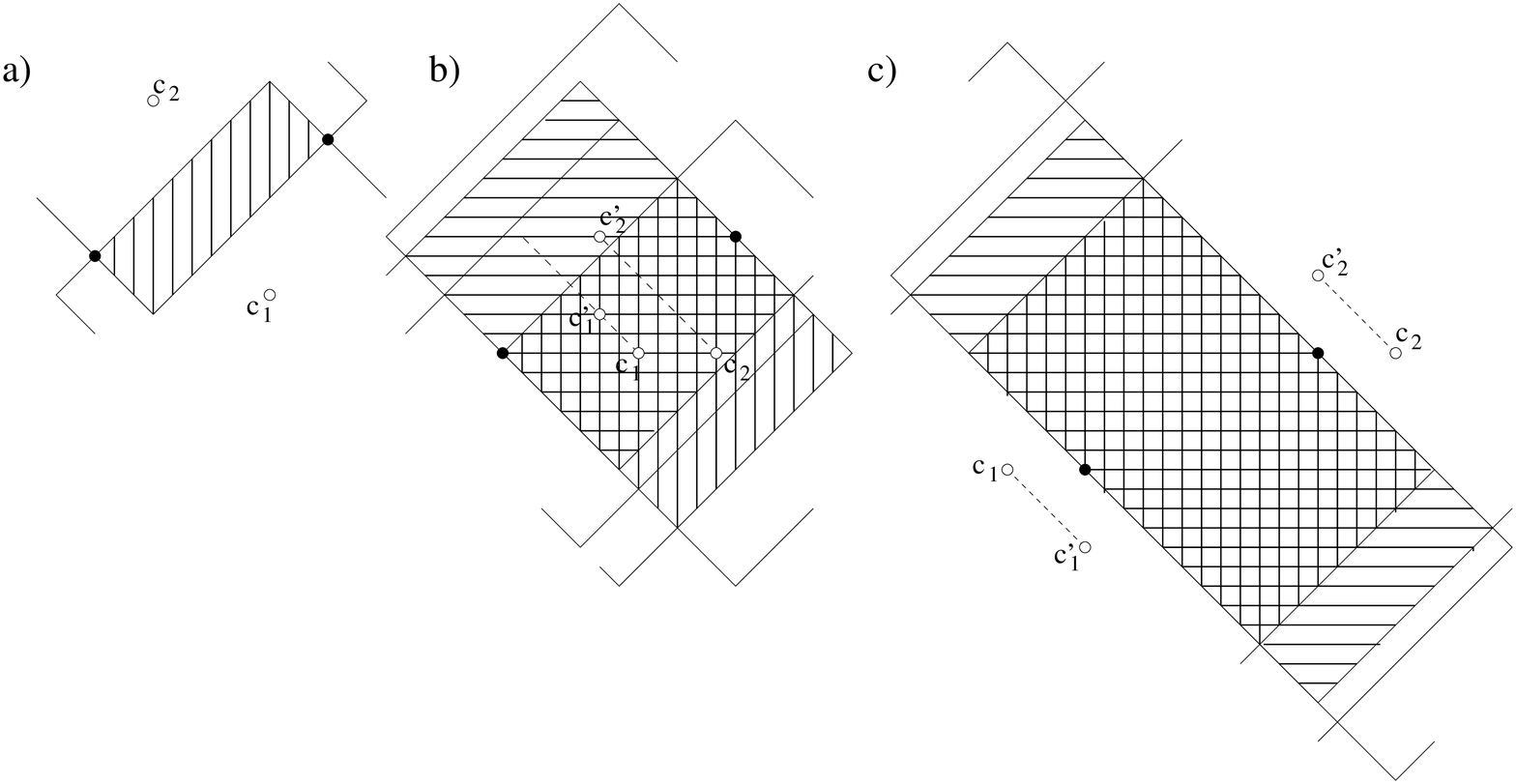}
\caption{Lenses in $l_1$ metric: (a) $0 < \beta < 1$, (b) $1 <\beta < 2$,
(c) $2 < \beta < \infty$.}
\label{fig:l_1}
\end{figure}

\begin{corollary}
\label{equality}
For a given set of points $V$ and the parameter $0 \leq \beta \leq \infty$ the graph 
$G_{\beta}^1(V)$ is a set of edges such that $v_1v_2 \in G_{\beta}^1(V)$ if and only if 
there exists a lens $N_1(v_1, v_2,\beta)$ that no point from $V \setminus \{v_1, v_2\}$ belongs 
to it. $G_{\beta}^1(V)$ is called $\beta$-skeleton in $l_1$.
For any $\beta < 1$ we have $G_{\beta}^1(V)=G_{0}^1(V)$ and for any $\beta \geq 2$ there is 
$G_{\beta}^1(V)=G_{2}^1(V)$.
\end{corollary}

Finally, we can prove that:

\begin{lemma}
\label{finalequalityl1}
For $1 \leq \beta \leq \beta' \leq 2$ in $l_1$ metric we have:\\
$MST(V) \subseteq G_{\beta'}^1(V) \subseteq G_{\beta}^1(V) \subseteq DT(V)$.
\end{lemma}
\begin{proof}  
Let $c_1, c_2$ be centers of discs determining a lens $N_1(v_1,v_2,\beta)$.
Let $c_1'$ be a point such that $d_1(v_1,c_1')=\frac{\beta'}{\beta}d_1(v_1,c_1)$
and $d_1(v_1,c_1')+d_1(v_2,c_1')=d_1(v_1,v_2)$
and $x(c_1')=x(c_1)$ or $y(c_1')=y(c_1)$ (one of such choices is always possible).
Then disc with center in $c_1'$ and radius $d_1(v_1,c_1')$ contains a disc with center
in $c_1$ and radius $d_1(v_1,c_1)$. 
In the same way we define a disc center $c_2'$.  
Then $c_1', c_2'$ are centers of discs determining a lens $N_1(v_1,v_2,\beta')$
and $N_1(v_1,v_2,\beta) \subseteq N_1(v_1,v_2,\beta')$. 
Hence $G_{\beta'}^1(V) \subseteq G_{\beta}^1(V)$. \\
The other inclusions can be proved in the same way as in the paper by Kirkpatrick and Radke 
\cite{kr85}.    
\end{proof} 

\begin{corollary}
Let $GG_1(V)$ and $G_{MG}(V)$ be Gabriel graphs defined in \cite{w06} and \cite{mg11}, respectively.
Then $GG_1(V) \subseteq G_{MG}(V) \subseteq G_1^1(V)$.
\end{corollary}

\section{$\beta$-skeletons for weighted graphs}

Let $G=(V,U,E)$ be a connected, edge-weighted graph in the plane, where $U \subseteq V$ 
is a subset of the set of vertices of $G$. Each edge of $e \in E$ is labelled with 
a positive weight $w(e)$. 
The distance $d_G(p,q)$ between two points $p$ and $q$ of $G$ is the minimum total weight 
of any path connecting $p$ and $q$ in $G$.
Graph $G$ with the distance function $d_G$ is a metric space.\\
The closed disc $D_G(p,r)$ is defined as the set of points $q$ of $G$ for which $d_p(p,q) \leq r$.\\

Abrego et al. \cite{a12} defined Minimum Spanning Tree and Delaunay Graph of $G$ 
in the following way:

\begin{definition}
A minimum spanning tree of $G=(V,U,E)$ is a tree $T=(U,F)$ such that the sum of $d_G(u_i,u_j)$ 
over all edges $u_iu_j \in F$ is minimal. 

The Delaunay graph of $G=(V,U,E)$, denoted by DG(U), is the graph $H=(U,F)$ such that 
$(u_i,u_j) \in F$ if and only if there exists a closed disc $D_G(p,r)$, where $p$ is a 
point of $G$, containing $u_i$ and $u_j$ and no other vertex from $U$.

\end{definition}

Let us assume (like \cite{a12}) that the shortest paths between each pair vertices in $U$ is uniquely 
defined. Similarly as in the previous Section we can define lenses $N_w(u_1,u_2,\beta)$ for edges 
of the graph $G$.
Unfortunately, lenses for $\beta < 1$ can be defined only for particular kinds of graphs
(for each edge $u_1u_2$ connecting vertices in $U$ a cycle of $d_G(u_1,u_2)(1+\frac{1}{\beta})$ 
length has to occur). Since the total weight of edges is limited lenses for a big value 
of $\beta$ do not exist either. Therefore we can define a $\beta$-skeleton $G_{\beta}^w(U)$
like in the Section \ref{l1-section} only in a limited range.

\begin{lemma}
\label{beta-weight}
Let $d_G(c(u_1u_2))$ denote the length of the shortest circle containing the edge $u_1u_2$.
The $\beta$-skeletons $G_{\beta}^w(U)$ for weighted graphs are correctly defined 
for $1 \leq \beta \leq \min_{u_1,u_2 \in U} \frac{d_G(c(u_1u_2))}{2d_G(u_1u_2)}+1$.
\end{lemma}   
\begin{proof}
A distance between centers of discs determining a lens $N_w(u_1,u_2,\beta)$ should be
$(\beta-1)d_G(u_1u_2)$. Hence, $d_G(c(u_1u_2)) \geq 2(\beta-1)d_G(u_1u_2)$, i.e.
$\beta \leq \frac{d_G(c(u_1u_2))}{2d_G(u_1u_2)}+1$. \\
Then for $1 \leq \beta < \beta' \leq \min_{u_1,u_2 \in U} \frac{d_G(c(u_1u_2))}{2d_G(u_1u_2)}+1$
and any $z_1,z_2 \in U$ if a lens $N \in N_w(z_1,z_2,\beta)$ then there exists
a lens $N' \in N_w(z_1,z_2,\beta')$ such that $N \subseteq N'$. 
\end{proof}

Note that $\frac{d_G(c(u_1u_2))}{2d_G(u_1u_2)}+1 \geq \frac{2d_G(u_1u_2)}{2d_G(u_1u_2)}+1=2$.
Hence $GG(U)$ and $RNG(U)$ are correctly defined.

\begin{lemma}
For $1 \leq \beta < \beta' \leq 2$ and for every graph $G=(V,U,E)$ the following inequalities occur: \\
$MST(U) \subseteq RNG(U) \subseteq G_{\beta'}^w{(U)} \subseteq G_{\beta}^w{(U)} \subseteq GG(U) \subseteq DG(U)$, \\
where $RNG(U)=G_{2}^w{(U)}$ and $GG(U)=G_1^w{(U)}$.
\end{lemma}
\begin{proof}
According to Lemma \ref{beta-weight} inclusions 
$RNG(U) \subseteq G_{\beta}^w{(U)} \subseteq G_{\beta'}^w{(U)} \subseteq GG(U) \subseteq DG(U)$
are true. The other inclusions were proved by Abrego et al. \cite{a12}. 
\end{proof}

\begin{figure}[htbp]
\centering
\includegraphics[scale=0.3]{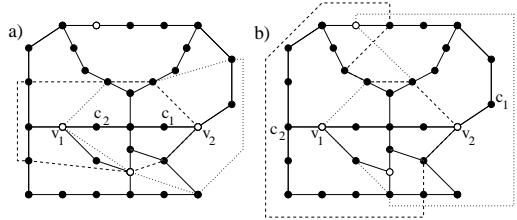}
\caption{Lenses in the weighted graph (each edge has the same weight) for different
values of $\beta$.}
\label{fig:weighted}
\end{figure}

\section{Not uniquely defined generators of lenses} 

In previous sections, we considered $\beta$-skeletons such that centers of discs determining lenses
were defined depending on position and distance of points from $V$. Those points generated sets
of lenses. What will happen when we replace points with objects, e.g. clusters of points or segments ?
We will study this problem focusing on the case of segments. \\ 

Let us consider a set $S$ of $n$ segments and the $l_p$ metric, where $1 \leq p \leq \infty$.
For two segments $s_1,s_2 \in S$ and a parameter $0<\beta<\infty$ we define $N_p^s(s_1,s_2,\beta)$
as a set of lenses $N_p(v_1,v_2,\beta)$, where $v_1,v_2$ are points such that $v_1 \in s_1$ 
and $v_2 \in s_2$. 

\begin{definition}
\label{betaskeleton3}
For a given set of segments $S$ and given parameters $0 \leq \beta$ and $0 \leq p \leq \infty$ 
we define a graph $G_{\beta}^s(S)$such that an edge between segments $s_1$ and $s_2$ exists 
if and only if there exists a lens in $N_p^s(s_1, s_2,\beta)$ whose intersection
with $S \setminus \{s_1,s_2\}$ is empty. 
\end{definition}

In order to define Delaunay triangulation we will use a definition by Brevilliers et al. \cite{bcs08}:
\begin{definition}
A segment triangulation $\mathcal{T}$ of $S$ is such a partition of the convex hull $conv(S)$ of $S$ 
in disjoint sites, edges and faces that: 
\begin{enumerate}
\item Every face of $\mathcal{T}$ is an open triangle whose vertices are in three distinct sites 
of S and whose open edges do not intersect S,
\item No face can be added without intersecting another one,
\item The edges of T are the (possibly two-dimensional) connected components of $conv(S)\setminus 
(F \cup S)$, where  $F$ is the union of faces of $\mathcal{T}$.
\end{enumerate} 
\end{definition}

A segment triangulation of $S$ is Delaunay $DT(S)$ if the circumcircle of each face does not contain 
any point of $S$ in its interior (see Figure \ref{fig:segments}). 

\begin{figure}[htbp]
\centering
\includegraphics[scale=0.2]{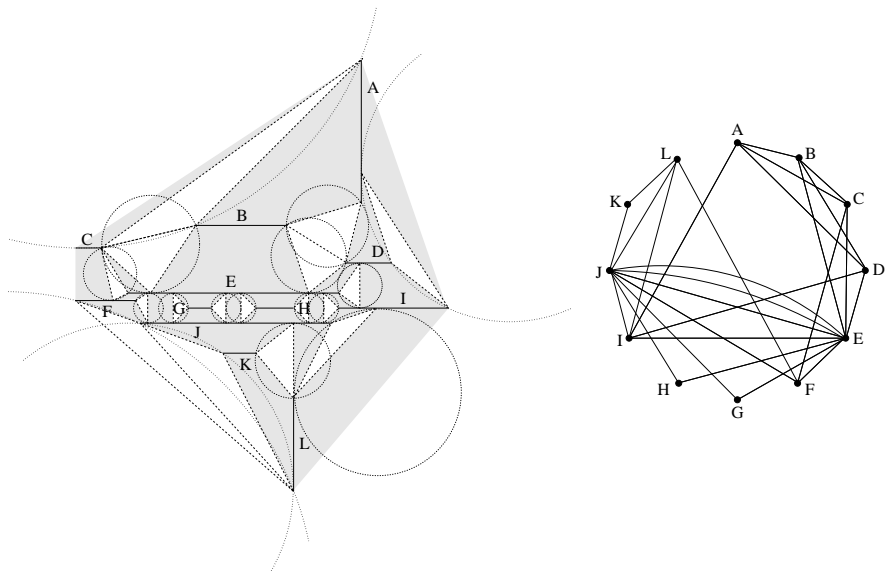}
\caption{Delaunay triangulation for a set of segments and a multigraph corresponding to it.}
\label{fig:segments}
\end{figure}

Note that we can present this triangulation as a multigraph $M=(V,E)$ with a set of vertices $V=S$, 
and a separate edge between vertices $s_1$ and $s_2$ for each edge in $DT(S)$ connecting segments 
$s_1$ and $s_2$ (see Figure \ref{fig:segments}). Each such edge in the graph can be labelled with 
the length of the shortest path between two points of a given edge in $DT(S)$ belonging to opposite 
segments.

We can formulate the following lemma.

\begin{lemma}
\label{gg-dt-s}
For $1 \leq \beta \leq \beta' \leq 2$ we have: \\ 
$MST(S) \subseteq RNG(S) \subseteq G_{\beta'}^s{(S)} \subseteq G_{\beta}^s{(S)} \subseteq GG(S) \subseteq DT(S)$,\\
where $RNG(S)=G_{2}^s{(S)}, GG(S)=G_1^s{(S)}$ and $MST(S)$ is a minimum spanning tree for a set 
of segments $S$.
\end{lemma}
\begin{proof} 
We want to show that $GG(S) \subseteq DT(S)$.
Let $v_1 \in s_2, v_2 \in s_2$ be such a pair of points that there is a $l_p$ disc $D$ 
with diameter $v_1v_2$ containing no points from segments from $S \setminus \{s_1, s_2\}$ inside of it. 
We transform $D$ by homothety in respect of $v_1$ so that its image $D'$ would be tangent to $s_2$ 
in $t$. Then we transform $D'$ by homothety in respect of $t$ so that its image $D''$ would be 
tangent to $s_1$ (see Figure \ref{fig:gg-dt}). The disc $D''$ lies inside of $D$, i.e. it does not 
intersect segments from $S \setminus \{s_1, s_2\}$, and is tangent to $s_1$ and $s_2$. Hence, 
if the edge $s_1s_2$ belongs to $GG(S)$ then it also belongs to $DT(S)$. \\
Let $c_1, c_2$ be centers of discs determining a lens $N \in N_p(s_1,s_2,\beta)$.
Let $c_1'$ ($c_2'$, respectively) be an image of $c_1$ ($c_2$, respectively)
by homothety with the factor $\frac{\beta'}{\beta}$ in respect of point $v_1$ ($v_2$, respectively).
Then $c_1', c_2'$ are centers of discs determining a lens $N' \in N_p(s_1,s_2,\beta')$
and $N \subseteq N')$. Hence $G_{\beta'}^s(V) \subseteq G_{\beta}^s(V)$. \\
The last inclusion can be proved in the same way as in the paper by Kirkpatrick and Radke 
\cite{kr85}.   

\begin{figure}[htbp]
\centering
\includegraphics[scale=0.35]{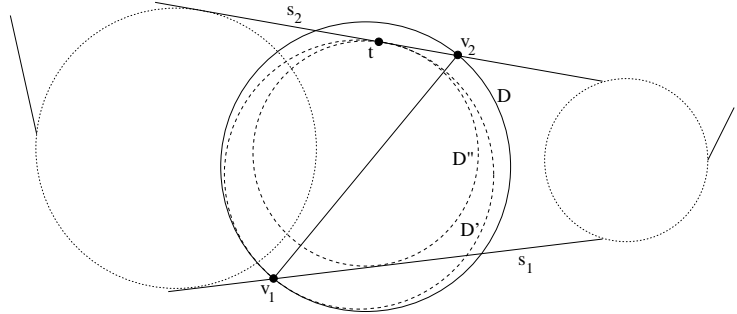}
\caption{$G_{\beta}^s(S) \subseteq DT(S)$.}
\label{fig:gg-dt}
\end{figure}  
\end{proof}

\section{Generalized $\beta$-skeletons}

Finally, we can formulate the definition of the generalized $\beta$-skeletons:

\begin{definition}
\label{betaskeletons5}
For a given set of objects $S$ in space $X$ with a metric $d$ and for parameters 
$0 \leq \beta \leq \infty$ we define a graph $G_{\beta}(S)$ - called generalized $\beta$-skeleton -
as follows: 
two objects $s_1$ and $s_2$ are connected with an edge if and only if  
at least one lens in $N_d(s_1,s_2,\beta)$ is not intersected by any object from 
$S \setminus \{s_1, s_2\}$ where:
\begin{itemize}
\item
for $0<\beta<1$ a lens $N_d(v_1,v_2,\beta) \in N_d(s_1,s_2,\beta)$, where $v_1 \in s_1$ 
and $v_2 \in s_2$, is the intersection of two discs, each of them has radius 
$\frac{d(v_1,v_2)}{2 \beta}$, whose boundaries contain both $v_1$ and $v_2$ and each shortest
path connecting their centers intersects some shortest path between $v_1$ and $v_2$, 

\item
for $1 \leq \beta<\infty$ a lens $N_d(v_1,v_2,\beta) \in N_d(s_1,s_2,\beta)$, where $v_1 \in s_1$ 
and $v_2 \in s_2$, is the intersection of two discs with centers $c_1$ and $c_2$, respectively,
such that $d(c_1, c_2)=(\beta -1)d(v_1,v_2)$ and 
$d(v_1, c_2)=d(v_2,c_1)=|\frac{(\beta -2) d(v_1v_2)}{2}|$.

\item
for  $\beta=0$ ($\beta=\infty$, respectively) a lens $N_d(v_1,v_2,0) \in N_d(s_1,s_2,0)$, 
($N_d(v_1,v_2,\infty) \in N_d(s_1,s_2,\infty)$, respectively), where $v_1 \in s_1$ 
and $v_2 \in s_2$, we have  
$N_d(v_1,v_2,0)=\lim_{\beta \rightarrow 0} N_d(v_1,v_2, \beta)$ and 
$N_d(v_1,v_2, \infty)=\lim_{\beta \rightarrow \infty} N_d(v_1,v_2, \beta)$.
 
\end{itemize}
\end{definition}

Note that in the similar way we can modify the definition of circle-based $\beta$-skeleton.
Moreover, the definition (in both cases) can be applied in multidimensional spaces.  
For example, for $0 < \beta < 1$ in $R^n$, where $n \geq 2$, regions $N(v_1,v_2,\beta)$ are 
determined by spheres which centers are located symmetrically in respect of the edge $v_1v_2$.

\section{Algorithms}

In this section we will describe algorithms computing $\beta$-skeletons in above considered
situations.\\

We will start from $\beta$-skeletons in $R^2$ with $l_1$ metric. 
According to Lemma \ref{l1-01} for $0 < \beta < 1$ the lenses are unique defined. Moreover, 
the points from $V$ are their vertices. \\
We rotate plane to obtain axis-aligned lenses. Then we use a plane sweep algorithm.
We sweep from left to right. The event structure contains a $x$-ordered set $V$.
The state structure contains for each $w \in V$ $y$-ordered, labelled by $w$, list of vertices 
lying on the right from $v$. When the sweep line reaches a vertex $v$ we remove from each list being 
in the state structure all points lying on the opposite side of $v$ than a vertex labelling the list. 
If $v$ is on the list we remove them and we add an edge connecting $v$ with a vertex being a label 
of the list to the set of solutions (see Algorithm \ref{alg-1}). \\

\begin{algorithm}[H]
\footnotesize{
\SetKwFunction{Front}{Front}
\SetKwFunction{SEARCH}{SEARCH}
\SetKwFunction{Y}{Y}
\SetKwFunction{return}{return}
\KwIn{rotated plane and a set of points $V$}
\KwOut{a set $E$ of all edges of $\beta$-skeleton for $V$ }
\BlankLine
create list $X=\{x_1, \ldots, x_n\}$ of all vertices sorted by their $x$-coordinate from left to right\;

\For{all vertices $p \in V$}{
create a list $\Y(p)$ of vertices that follow $p$ in the $X$ list, sorted by their $y$-coordinates\;
}
\For{$i=2$ to $n$}{
    \For{$j=1$ to $i-1$}{
    remove from the list $\Y(x_j)$ all vertices that are separated from $x_j$ on the this list by $x_i$\;
     \If{the first not deleted element on the list $\Y(x_j)$ is $x_i$}{
       add $(x_i,x_j)$ to $E$\;
       remove $x_i$ from $\Y(x_j)$\;
      } 
     }
}
}
\caption{Algorithm computing $G_{\beta}^1(V)$ for $0 < \beta < 1$ }
\label{alg-1} 
\end{algorithm}  

\begin{theorem}
For a given $0<\beta<1$ and a set $V$ of n points in $R^2$ with $l_1$ metric the $\beta$-skeleton
$G_{\beta}^1(V)$ can be computed in $O(n^2)$ time.
\end{theorem}
\begin{proof}
We can rotate and sort points of $V$ in $O(n \log n)$ time. Lists $Y(x_i)$ can be prepared 
in total $O(n^2)$ time. Time needed for loops activity is $\Sigma_{v \in V} O(k_i+1)$, where
$k_i$ is a number of deleted elements in $i$-th turn of the first loop. 
Since $\Sigma_{v \in V} O(k_i)=O(n^2)$, the algorithm complexity is $O(n^2)$.
\end{proof}

Based on Lemmas \ref{l1-twolenses} and \ref{l1-betabig} it follows that if an edge $v_1v_2$
belongs to a $\beta$-skeleton for $1 \leq \beta < \infty$ then there exist a lens $N$
such that $N \cap V \setminus \{v_1,v_2\} = \emptyset$ and $N \subseteq N_1(v_1,v_2,2)$.
Moreover, for $1 \leq \beta < 2$ the lens $N$ is one of two extremal lenses
or it lies between points eliminating those lenses. \\
Let us rotate plane to obtain axis-aligned lenses. Let us consider a lens such that points of $V$
generating it lie on the top and bottom sides of the lens. To solve the problem, it is sufficient 
to find leftmost point in rightmost lens and rightmost point in leftmost lens. Then we can verify
if the distance between those points is big enough. Note that any point from $V$ inside
$N_1(v_1,v_2,2)$ eliminates an edge $v_1v_2$ from $\beta$-skeleton for $2 \leq \beta < \infty$.
According to Lemma \ref{finalequalityl1} a $\beta$-skeleton is a subset of $DT(V)$.
Hence we have to analyze only $O(n)$ lenses. \\
We use a plane sweep algorithm. We will sweep a plane to the right analyzing rightmost lenses
and next to the left analyzing leftmost ones. The event structure contains positions of $V$ 
elements and sides od lenses sorted in respect of $x$-coordinate. The state structure is 
an interval tree with $n$ leaves corresponding to points of $V$. In nodes we will save 
information about swept lenses. If a sweep line visit a lens we mark a node being a lowest 
ancestor of leaves corresponding to the top and bottom sides of analyzed lens 
(see Figure \ref{fig:l1-alg}). 
If a sweep line reaches a point $v \in V$ we find a path from a root of the tree to a leaf 
corresponding to $v$. We allocate the point to all lenses marked on this path. Then we erase
the markers. If the sweep line leaves the lens we erase the corresponding marker from the tree
(if the marker exists).  

\begin{figure}[htbp]
\centering
\includegraphics[scale=0.4]{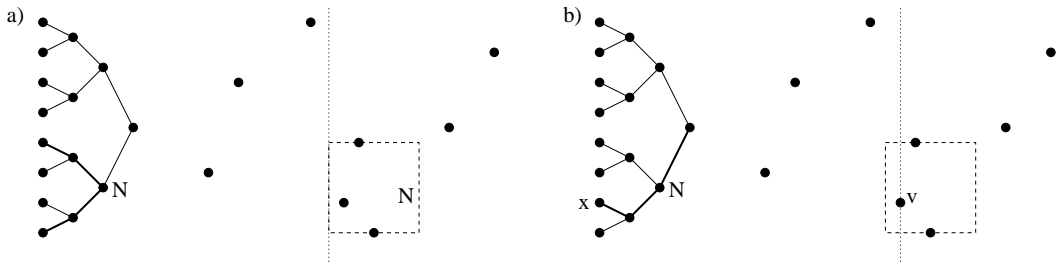}
\caption{A plane sweep algorithm for $G_{\beta}^1(V)$, where $1 \leq \beta < \infty$.
sweep line ecounters a lens side (a) or a point of $V$ (b). }
\label{fig:l1-alg}
\end{figure}

\begin{algorithm}[H]
\footnotesize{
\SetKwFunction{Front}{Front}
\SetKwFunction{SEARCH}{SEARCH}
\SetKwFunction{p}{p}
\SetKwFunction{return}{return}
\KwIn{rotated plane with rightmost lenses and a set of points $V$}
\KwOut{leftmost points belonging to rightmost lenses}
\BlankLine

\While{there are still unvisited points from $V$ or lens sides }{
sweep the plane from left to right\;
\If{ecounter a lens side}
{ mark a node in the interval tree corresponding to the side of the lens\;
 }
\If{ecounter a point $v$}
{ find a path from the root of tree to the leaf corresponding to $v$\;
 \For{each lens mark $N$ on the path}
 {save the pair $(v,N)$ and erase $N$ from the tree }
}
\If{leave a lens}
{ remove mark $N$ (if it exists) from the tree \;
}
}
}
\caption{Algorithm computing leftmost points belonging to rightmost lenses}
\label{alg-2}   
\end{algorithm}   

\begin{theorem}
For a given $1 \leq \beta < \infty$ and a set $V$ of $n$ points in $R^2$ with $l_1$ metric 
the $\beta$-skeleton $G_{\beta}^1(V)$ can be computed in $O(n \log n)$ time.
\end{theorem}
\begin{proof}
We use a plane sweep algorithm four times (two times in each direction after clockwise 
and counterclockwise rotation of the plane - depending on a slope of edges generating lenses).
The number of events is $O(n)$. The $i$-th event needs $O(\log n + k_iO(\ log n))$ time, where 
$k_i$ is a number of added or removed marks. Since $\Sigma k_i = O(n)$, the algorithm complexity is 
$O(1) \times (O(n \log n) +O(n \log n)) = O(n \log n)$
\end{proof}

Now, we will present an algorithm computing $G_{\beta}^w(U)$ for a weighted graph $G=(V,U,E)$.\\
The algorithm is natural. First we compute distances between all vertices in $G$ using 
e.g. Johnson's algorithm \cite{clrs09}. Then we find a set of points which are candidates for centers
of disc determining lenses. For this purpose we analyze distances of two points generating lens
from point belonging to any edge. Next, we verify which pairs of candidates can create lenses
(a ratio of a distance between centers of discs determining lens to a distance between generators
is $(\beta-1)$). In the last step analyze an intersection of a set $U$ with each lens. \\
Unfortunately, the algorithm is expensive due to a big number of possible lenses.

\begin{lemma}
A number of all candidates for centers of disc determining lenses is $O(m^2)$, where $m = |E|$.
\end{lemma}
\begin{proof}
Let us assume that $|V|=n$, $n$ is even, $V = C_1 \cup C_2$, $|C_1|=|C_2|$, sets $C_1$ and $C_2$ 
form cycles which every second element belongs to $U$. One edge of $C_1$ has a big weight.
Edges connecting vertices of $C_1 \setminus U$ with vertices of $C_2$ also have
a big weight (see Figure \ref{fig:sr-wag}). \\
Hence the graph $G$ has $m=O(n^2)$ edges.
For a sufficiently big value of parameter $\beta$ and each of $O(n-k)$ pairs vertices generating 
a lens, which are connected by $2k$ short edges, there exist $O(n-k) \times O(n)$ edges containing 
candidates for disc centers. Hence, the total number of candidates is 
$\Sigma_{k=1}^{n-1} O(n-k) \times O(n-k) \times O(n) = O(n^4)=O(m^2)$.
\end{proof}

\begin{figure}[htbp]
\centering
\includegraphics[scale=0.4]{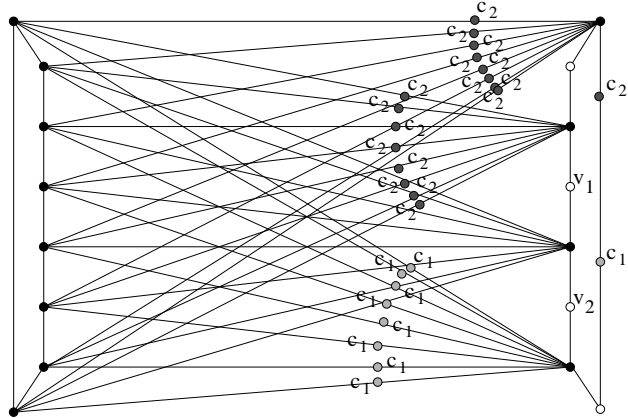}
\caption{ An example of many centers of disc determining lenses for generators $v_1, v_2$.}
\label{fig:sr-wag}
\end{figure}

\begin{algorithm}[H]
\footnotesize{
\SetKwFunction{Front}{Front}
\SetKwFunction{SEARCH}{SEARCH}
\SetKwFunction{A}{A}
\SetKwFunction{PotCen}{PotCen}
\SetKwFunction{return}{return}
\KwIn{weighted graph $G=(V,U,E),$ parameter $\beta \geq 1$}
\KwOut{set $F$ of all edges of $\beta$-skeleton for $G$ }
\BlankLine
compute all distances in graph $G$ between vertices in $V$\;
\For{every edge $u_1u_2$}{
   find all potential centers $\PotCen(u_1,u_2)$ for $u_1u_2$\;
     \For{every pair $c_1,c_2 \in \PotCen(u_1,u_2)$}{
       \If{$d_G(c_1,c_2)=(\beta -1)d_G(u_1,u_2)$}{
       add pair $(c_1,c_2)$ to $\A(u_1,u_2)$\;
          }
       }
     \For{every pair $(c_1,c_2) \in \A(u_1,u_2)$}{
      \If{lens defined by $c_1,c_2$ is empty}{
      add $u_1u_2$ to $F$\;
      }
      }
}
}
\caption{Algorithm for computing $\beta$-skeleton for $\beta \geq 1$ for weighted graphs }  
\end{algorithm}  

\begin{theorem}
Let $d_G(c(u_1u_2))$ denote the length of the shortest circle containing the edge $u_1u_2$ 
in the weighted graph $G=(V,U,E)$.
For $1 \leq \beta \leq \min_{u_1,u_2 \in U} \frac{d_G(c(u_1u_2))}{2d_G(u_1u_2)}+1$
the $\beta$-skeleton $G_{\beta}^w(U)$ can be computed in $O(nm^3)$ time, where $|V|=n$ and $|E|=m$.
\end{theorem}
\begin{proof}
The most time-consuming is the last part of the algorithm, i.e. verification of emptiness of lenses.
For one pair of generators there are at most $O(m)$ candidates for centers of disc determining lenses.
Hence, there are $O(m) \times O(m^2)= O(m^3)$ lenses to verify. The total algorithm complexity
is $O(nm^3)$.
\end{proof}

\begin{corollary}
For $\beta=2$ the $\beta$-skeleton $G_{\beta}^w(U)$ can be computed in $O(n^3)$ time.
\end{corollary}
\begin{proof}
Centers of disc determining lenses for $RNG$ are uniquely defined. There are points of $U$.
Since $|U|=O(n)$, there are at most $O(n^2)$ lenses. Complexity of the Johnson's algorithm
id $O(nm)$. Hence, the total complexity of the algorithm is $O(nm)+O(n^3)=O(n^3)$.   
\end{proof}

The most interesting and difficult is an algorithm computing $\beta$-skeletons for a set $S$
of $n$ segments in $R^2$ with $L_2$ metric. We will outline the solution. 
Details can be found in the paper \cite{km14}.
Let us consider a set of parametrized lines containing given segments. A line $P(s_i)$
contains a segment $s_i \in S$ and its parametrization is 
$(x_1^i,y_1^i)+t_i \times [x_2^i-x_1^i,y_2^i-y_1^i]$, where $(x_1^i,y_1^i)$ and $(x_2^i,y_2^i)$ are
ends of the segment $s_i$. Let $s_1$ and $s_2$ be generators of a lens.
We shoot a rays from points $q(t_1) \in P(s_1)$. Let us assume that the rays pass through
a point $h \in s$, where $s \in S \setminus \{s_1,s_2\}$.
For $0 < \beta < 1$ the ray reflects. The sum of angle of incidence and angle of reflection 
is equal to an inscribed angle for a lens generated for a given value of $\beta$.
For $1 \leq \beta < \infty$ we compute a new line perpendicular to the ray such that 
a distance between an intersection point $r$ and the point $q(t_1)$ is 
$d(p,q(t_1))=\frac{1}{\beta}d(h,q(t_1))$.
Both new created lines will be called reflected rays.
The reflected ray intersects line $P(s_2)$ in point $w(t_2)$.
A relation between parameters $t_1$ and $t_2$ is a hyperbolic function.
For a given segment $s \in S \setminus \{s_1,s_2\}$, hyperbolas for points $h \in s$
create hyperbolic stripe. Intersection of this stripe with a hyperbolic stripe describing relation
between parameters $t_2$ and $t_1$ is a polygon whose sides are parts of hyperbolas 
(see Figure \ref{fig:rays}). 

\begin{figure}[htbp]
\centering
\includegraphics[scale=0.3]{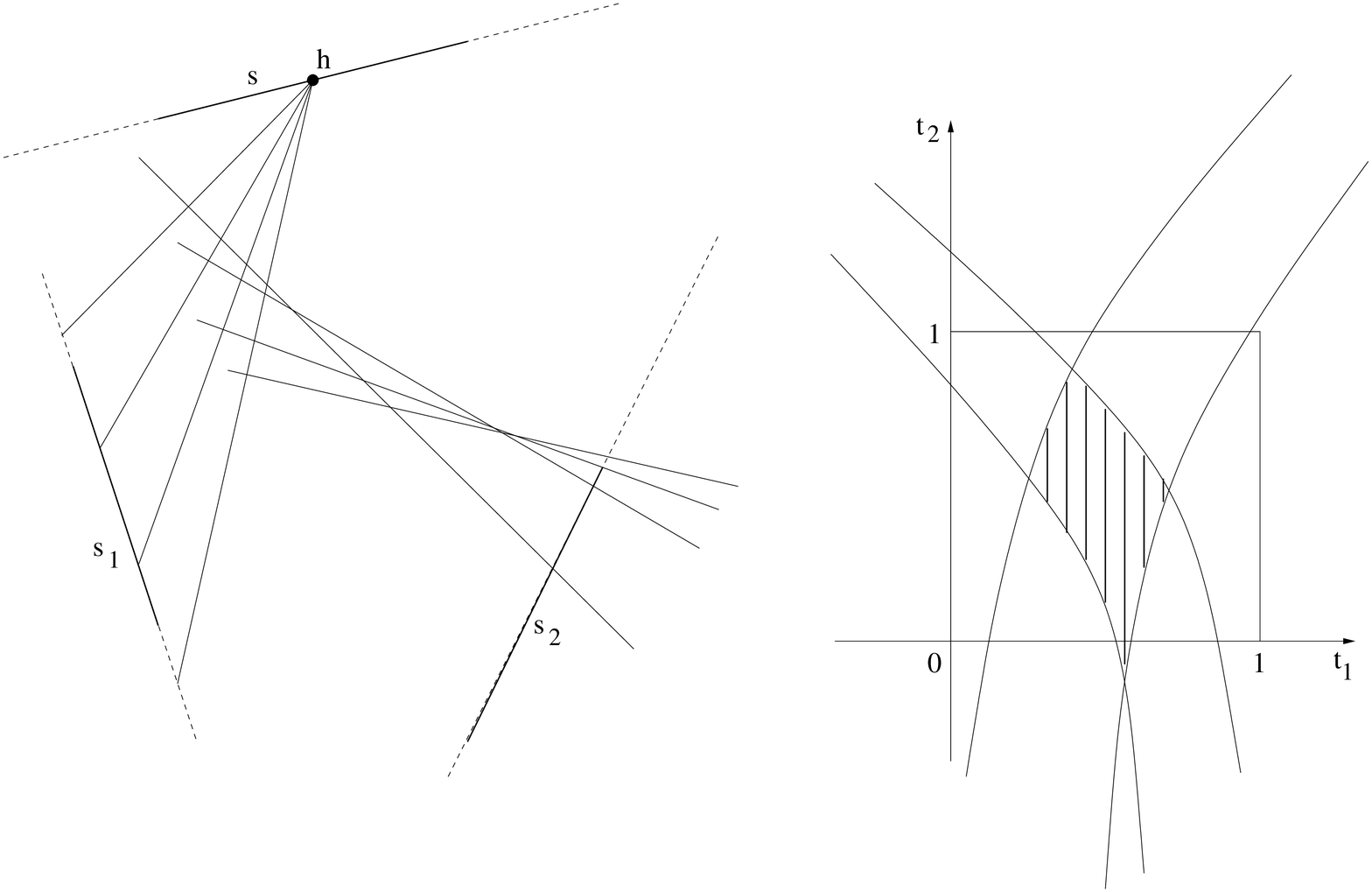}
\caption{Rays for computing of $RNG$ (a), hyperbolic polygon (b).}
\label{fig:rays}
\end{figure}   

\begin{theorem}
An edge generated by segments $s_1$ and $s_2$ belongs to $\beta$-skeleton $G_{\beta}^s(S)$
if and only iff a sum of polygons computed for all segments $s \in S \setminus \{s_1,s_2\}$
does not cover a square $[0,1] \times [0,1]$. 
\end{theorem}
\begin{proof}
If there exists an uncovered point $(t_1,t_2) \in [0,1] \times [0,1]$, then a lens generated 
for points $q(t_1)$ and $w(t_2)$ does not intersect any segment in $S \setminus \{s_1,s_2\}$.
The opposite implication also is true.
\end{proof}

\begin{theorem}
The $\beta$-skeleton $G_{\beta}^s(S)$ for a set of $n$ segments $S$ in $R^2$ with $L_2$ metric
can be computed for $0 < \beta < 1$ in $O(n^4)$ time and for $1 \leq \beta < \infty$ in $O(n^3)$
time.
\end{theorem}
\begin{proof}
A sum of $n$ polygons can be computed in $O(n^2)$ time. A verification $O(n^2)$ pairs 
of lens generators needs $O(n^4)$ time.
According to Lemma \ref{gg-dt-s} and linearity of $DT(S)$, for $1 \leq \beta < \infty$
we have to check only $O(n)$ pairs of generators. Hence, a complexity of the algorithm
is $O(n^3)$. 
\end{proof}

For $\beta=1$ we can use $2$-order Voronoi diagrams for segments to compute $GG(S)$.
In this case a complexity of the algorithm is $O(n \log n)$ \cite{km14}.

\section{Conclusions}

In this paper we show a way of defining $\beta$-skeletons in general.
We have based our proposition on a distance criterion and we described conditions
should be satisfied if the lenses are not defined uniquely.
We have focused our considerations only on a few special cases which in our opinion
well describe the idea of this general definition.
In a similar way, we can also define $\beta$-skeletons for example for a set of polygons
or for the jungle river metric.
It is also easy to generalize this definition for higher dimensions. 
One can consider a couple of new problems regarding this definition.
It would be interesting to check how those changes can influence the time 
of algorithms computing $\beta$-skeletons. For example, if the $RNG$ for segments could be
computed faster than in $O(n^3)$ time.
Many questions can also concern properties of $\beta$-skeletons which they have for different objects
and their practical applications. 
 
\noindent
{\bf Acknowledgements}\\
\noindent
The authors would like to thank Evanthia Papadopoulou for important discussions. 
We thank Jerzy W. Jaromczyk for his comments and suggestions.

\bibliographystyle{abbrv}

\end{document}